\pdfoutput=1

\documentclass{article}
\usepackage{graphicx}
\usepackage{spconf}
\usepackage{color}
\usepackage{float}
\usepackage{amsmath,bm,amssymb,amsfonts,amsthm,epsfig}
\usepackage{graphicx}
\usepackage{psfrag}
\usepackage[caption=false]{subfig}
\usepackage{comment}
\usepackage{url}
\usepackage{cite}

\graphicspath{{Figures/}}
\usepackage{epstopdf}
\usepackage{lipsum}

\newtheorem{proposition}{Proposition}
\newtheorem{thm}{Theorem}
\newtheorem*{thm*}{Theorem}

\newtheorem*{remark*}{Remark}
\usepackage{soul}
\usepackage[hidelinks]{hyperref}
\usepackage[capitalize]{cleveref}
\crefname{thm}{Theorem}{Theorems}

\newcounter{opteq}
\newenvironment{opteq}{\refstepcounter{opteq}\align}{\tag{P\theopteq}\endalign}
\usepackage{arydshln}
\usepackage{mathrsfs}
\usepackage{bbm}

\usepackage{pifont}

\usepackage[toc,page]{appendix}

\DeclareMathOperator*{\argmax}{arg\,max}
\DeclareMathOperator{\T}{\top}
\DeclareMathOperator{\HT}{H}

\DeclareMathOperator{\rank}{rank}

\makeatletter
\newcounter{savesection}
\newcounter{apdxsection}
\renewcommand\appendix{\par
	\setcounter{savesection}{\value{section}}%
	\setcounter{section}{\value{apdxsection}}%
	\setcounter{subsection}{0}%
	\gdef\thesection{\@Alph\c@section}}
\newcommand\unappendix{\par
	\setcounter{apdxsection}{\value{section}}%
	\setcounter{section}{\value{savesection}}%
	\setcounter{subsection}{0}%
	\gdef\thesection{\@arabic\c@section}}
\makeatother

\usepackage{enumitem}
\usepackage{tikz}
\usetikzlibrary{tikzmark,calc,decorations.pathreplacing,shapes,backgrounds,patterns}
\usepackage{pgfplots} 
\pgfplotsset{compat=1.16}
\pgfplotsset{table/search path={Data},
        colormap={parula}{
            rgb255=(53,42,135)
            rgb255=(15,92,221)
            rgb255=(18,125,216)
            rgb255=(7,156,207)
            rgb255=(21,177,180)
            rgb255=(89,189,140)
            rgb255=(165,190,107)
            rgb255=(225,185,82)
            rgb255=(252,206,46)
            rgb255=(249,251,14)
        },
        every axis/.append style={
                    label style={font=\small},
                    tick label style={font=\small},  
                    title style={font=\small}, 
                    },
    contour/label node code/.code={%
        \node{$m=\pgfmathprintnumber{#1}$};},
        log x ticks with fixed point/.style={
      xticklabel={
        \pgfkeys{/pgf/fpu=true}
        \pgfmathparse{exp(\tick)}%
        \pgfmathprintnumber[fixed relative, precision=2]{\pgfmathresult}
        \pgfkeys{/pgf/fpu=false}
      }
  }
    }
\usepgfplotslibrary{fillbetween}
\usetikzlibrary{pgfplots.groupplots}
\usepackage[percent]{overpic}

\newcommand\blfootnote[1]{%
  \begingroup
  \renewcommand\thefootnote{}\footnote{#1}%
  \addtocounter{footnote}{-1}%
  \endgroup
}
\newcommand{\rd}[1]{\textcolor[rgb]{0,0,0}{#1}}
\newcommand{\bl}[1]{\textcolor[rgb]{0,0,0}{#1}}
\newcommand{\ff}[1]{\textcolor[rgb]{0,0,0}{#1}}

\title{Effect of Beampattern on Matrix Completion with Sparse Arrays}
\name{\rd{Robin Rajam\"{a}ki, \ff{Mehmet Can H\"uc\"umeno\u{g}lu}, Pulak Sarangi, and \ff{Piya Pal}}}
\address{Department of Electrical and Computer Engineering, University of California San Diego}
\begin{document}

\maketitle

\begin{abstract}
We study the problem of noisy sparse array interpolation, where a large virtual array is synthetically generated by interpolating missing sensors using matrix completion techniques that promote low rank. The current understanding is quite limited regarding the effect of the (sparse) array geometry on the angle estimation error (post interpolation) of these methods. In this paper, we make advances towards solidifying this understanding by revealing the role of the {\em physical beampattern} of the sparse array on the performance of low rank matrix completion techniques. When the beampattern is analytically tractable (such as for uniform linear arrays and nested arrays), our analysis provides concrete and interpretable bounds on the scaling of the angular error as a function of the number of sensors, and demonstrates the effectiveness of nested arrays in presence of noise and a single temporal snapshot.

\blfootnote{This work was supported in part by grants ONR N00014-19-1-2256, NSF 2124929, and DE-SC0022165.} 
\end{abstract}

\section{Introduction}



Sparse sensor arrays offer several advantages over conventional uniform arrays, such as enhanced angular resolution \cite{sarangi2023superresolution} and the ability to identify more sources than sensors \cite{wang2017coarrays}. Consequently, they hold great promise in a plethora of emerging applications, including mmWave channel estimation \cite{haghighatshoar2018low}, automotive radar \cite{patole2017automotive}, and integrated sensing and communications (ISAC) \cite{liu2023integrated}. A key challenge is 
to fully leverage the benefits \rd{of sparse arrays} in sample-starved \rd{(even single-snapshot)} \cite{sarangi2021beyond,sarangi2022single,hucumenoglu2023toregularize,sarangi2023superresolution,sun2020asparse,liu2021rank,amini2022,ma2021multi} scenarios. \emph{Array interpolation} provides a potential remedy to these challenges \cite{sarangi2022single}. In particular, matrix completion based approaches have gained significant attention \cite{sun2020asparse,amini2022,sarangi2022single} \rd{due to their ability to} harness the {\em large filled aperture} of sparse arrays for  
high-resolution beamforming and direction-of-arrival estimation. 
Despite the surge of research on sparse array matrix completion, several theoretical questions remain open. Most notably, the impact of the physical array geometry on source localization error (post interpolation) is poorly understood. 
We address this outstanding issue by establishing that the  \emph{beampattern of sparse arrays} plays a major role in matrix completion based array interpolation. We rigorously \rd{show} that the {\em physical array beampattern} controls the worst-case error of these interpolation techniques, even when the goal is \emph{not to perform beamforming}, but to estimate target parameters using other methods. Since beamforming is a linear operation, 
it is not \emph{a priori} obvious that \rd{the beampattern} should also directly influence the error of low-rank matrix completion---a  \emph{nonlinear} and \emph{nonconvex} optimization problem. 
\bl{Although several nonlinear beamforming schemes \cite{vaidyanathan2011sparsesamplers,adhikari2017spatial,cohen2018sparseconvolutional} have been designed for sparse arrays,  
the connection between the (irregular) beampattern of sparse arrays and noisy matrix completion has not been investigated}. Our results also offer 
an intuitive interpretation of the worst-case angular error (specifically, its scaling with the number of sensors) of any given array geometry in terms of the associated beampattern and its main lobe width (spatial resolution) / side lobe levels (noise robustness) with respect to the signal-to-noise ratio (SNR).

To make the connection concrete and for ease of exposition, we focus on a single source model with an unknown angle and (positive) amplitude. Single source models are relevant in, e.g., beam alignment for mmWave communication \cite{chiu2019active}, as well as target detection and tracking for ISAC \cite[pp.~122, 422]{liu2023integrated}. The physical beampattern of sparse arrays will continue to play an important role for multiple (possibly real or complex-valued) sources. While such extensions are possible, they are not straightforward and are part of ongoing work.
\emph{Notation:} 
Given matrix $\bm{X}\in\mathbb{C}^{N\times M}$, $\bm{X}_{\mathbb{X}}$ denotes the $|\mathbb{X}|\times M$ matrix formed by retaining the $|\mathbb{X}|\leq N$ rows indexed by set $\mathbb{X}+1\subset\mathbb{N}_+$. The Hermitian Toeplitz matrix whose first column is $\bm{t}\in\mathbb{C}^{N}$ is denoted by $\mathcal{T}(\bm{t})\in\mathbb{C}^{N\times N}$.
\section{Signal model}\label{sec:signal_model}

Consider a $P$-sensor linear array with sensor positions given by set $\mathbb{D}=\{d_1,d_2,\ldots,d_P\}\subset \mathbb{N}$, $d_1=0<d_2<\ldots<d_P$ in units of half the carrier wavelength. We will focus on the uniform linear array (ULA), $\mathbb{D}=\mathbb{U}_P\triangleq \{0,1,\ldots,P-1\}$, and nested array \cite{pal2010nested}, $\mathbb{D}=\mathbb{S}\triangleq \mathbb{U}_{M}\cup((M+1)\mathbb{U}_{M}+M)$, where $M\triangleq\frac{P}{2}$ and $P$ is assumed even. 
Suppose we observe a single snapshot of a narrowband far field source signal with unknown angular direction $\theta\in[-\frac{\pi}{2},\frac{\pi}{2})$ and (positive) amplitude $\alpha>0$ impinging on the array in the presence of noise:
\begin{align}
    \bm{y} = \alpha \bm{a}_{\mathbb{D}}(\theta)+\bm{n}.\label{eq:y}
\end{align}
Here, $\bm{a}_{\mathbb{D}}(\theta)\in\mathbb{C}^P$ is the  steering vector satisfying $[\bm{a}_{\mathbb{D}}(\theta)]_i=e^{j\pi d_i\sin\theta}, d_i\in\mathbb{D}$ and noise is assumed \bl{to be} bounded as $\|\bm{n}\|_2\leq \epsilon$. 

 The objective of array interpolation is to emulate a large virtual ULA $\mathbb{U}_N\supseteq \mathbb{D}$ with $N\geq P$ sensors, by ``computationally filling in" the missing sensors. 
 A popular array interpolation approach is to solve a matrix completion problem \cite{abramovic1999positive,qiao2017unified,sun2020asparse,sarangi2022single,hucumenoglu2023toregularize} such as the following positive semidefinite (PSD) Toeplitz completion problem:
\begin{opteq}
\underset{\bm{t}\in\mathbb{C}^{N}}{\text{minimize}}\ \rank\mathcal{T}(\bm{t})\	\text{s.t.}\ \|\bm{y}-\bm{t}_{\mathbb{D}}\|_2\leq \epsilon, \mathcal{T}(\bm{t})\succeq 0. \label{p:psd}
\end{opteq}
%
In the following, we denote the steering vector of $\mathbb{U}_N$ by $\bm{a}(\theta)\in\mathbb{C}^N$, where $a_{i}(\theta)=e^{j\pi (i-1) \sin\theta}, i=1,2,\ldots,N$. Hence, $\bm{a}_{\mathbb{D}}(\theta)\in\mathbb{C}^P$ can be interpreted as consisting of a subset of the entries of $\bm{a}(\theta)$, sampled by $\mathbb{D}$.
\begin{proposition}[Rank-1 solutions]\label{thm:psd}
    Suppose $\bm{y}=\alpha\bm{a}_{\mathbb{D}}(\theta)+\bm{n}$, where $\|\bm{n}\|_2\leq \epsilon$ and $\alpha >2\epsilon /\sqrt{P}$. Then any solution $\bm{\hat{t}}\in\mathbb{C}^N$ to \eqref{p:psd} is of the form $\bm{\hat{t}}=\hat{\alpha}\bm{a}(\hat{\theta})$, $\hat{\alpha}>0$.
\end{proposition}
\begin{proof}
     Let $\bm{\hat{t}}$ be a minimizer of \eqref{p:psd}. Note that $\bm{t}'=\alpha\bm{a}(\theta)$ is a feasible point, since $\mathcal{T}(\bm{t}')=\alpha\bm{a}(\theta)\bm{a}^{\HT}(\theta)$ is PSD and $\|\bm{y}-\alpha\bm{a}_{\mathbb{D}}(\theta)\|_2=\|\bm{n}\|_2\leq \epsilon$. Hence, there exists a feasible point satisfying $\rank(\mathcal{T}(\bm{t}'))=1$. This is also the minimum rank solution, since a zero-rank solution is infeasible due to the assumption  
$\alpha>2\epsilon/\sqrt{P}$. 
      Finally, as $\mathcal{T}(\bm{\hat{t}})$ is a rank-1 PSD Toeplitz matrix, it admits Vandermonde decomposition $\mathcal{T}(\bm{\hat{t}}) = \hat{\alpha} \bm{a}(\hat{\theta}) \bm{a}^{\HT}(\hat{\theta})$ by Caratheodory's theorem \cite{cybenko1982moment}. Hence, $\bm{\hat{t}} = \hat{\alpha}\bm{a}(\hat{\theta}), \hat{\alpha}\!>\!0$.
\end{proof}
\cref{thm:psd} establishes that for large enough $\alpha$, any solution to \eqref{p:psd} will have the parametric form $\bm{\hat{t}}=\hat{\alpha}\bm{a}(\hat{\theta})$ which represents a \rd{scaled} virtual steering vector corresponding to angle $\hat{\theta}$. Consequently, the remainder of this paper focuses on the angle estimation error of the \emph{class of estimators that use $\hat{\theta}$ as an estimate of $\theta$}. Prominent members of this class include beamforming and subspace methods \cite{liao2016music,li2020super} that operate on $\bm{\hat{t}}$.\footnote{Beamforming seeks $\argmax_{\vartheta} |\bm{a}^{\HT}(\vartheta)\bm{\hat{t}}|\!=\!\argmax_{\vartheta} \hat{\alpha}|\bm{a}^{\HT}(\vartheta)\bm{a}(\hat{\theta})|\\ =\!\hat{\theta}$, 
whereas subspace methods find $\hat{\theta}$ by identifying the (rank-1) subspace spanned by $\bm{a}(\hat{\theta})$ from $\mathcal{T}(\bm{\hat{t}})=\hat{\alpha} \bm{a}(\hat{\theta})\bm{a}^{\HT}(\hat{\theta})$.} 
The following question underlies our main contribution: ``{\em \bl{Is} it possible to obtain a universal upper bound on the angle estimation error for this class?" }  
If so, what is the quantity of interest (specific to an array geometry) that determines such a bound?
Interestingly, an answer is given by the so-called \emph{unweighted physical array beampattern}. 

\section{Beampattern and angular error of matrix completion}\label{sec:beampattern}

The \emph{unweighted beamformer} 
of array $\mathbb{D}$ for spatial frequency $\omega\in\mathbb{R}$ is given by 
\begin{align}
    H_{\mathbb{D}}(\omega) \triangleq \sum_{d\in\mathbb{D}} e^{j \pi d\omega}. \label{eq:beampattern}
\end{align}
%
%
For any nonzero solution of \eqref{p:psd}, $\bm{\hat{t}}=\hat{\alpha}\bm{a}(\hat{\theta})\in\mathbb{C}^N$, we have 
\begin{align*}
    \|\bm{y}-\bm{\hat{t}}_{\mathbb{D}}\|_2=\|\alpha\bm{a}_{\mathbb{D}}(\theta)-\hat{\alpha}\bm{a}_{\mathbb{D}}(\hat{\theta})+\bm{n}\|_2=\|\bm{A}_{\mathbb{D}}\bm{\alpha}+\bm{n}\|_2,
\end{align*}
where $\bm{A}_{\mathbb{D}}\triangleq [\bm{a}_{\mathbb{D}}(\theta), \bm{a}_{\mathbb{D}}(\hat{\theta})]$ and $\bm{\alpha}\triangleq [\alpha,-\hat{\alpha}]^{\T}$. Hence,
\begin{align}
    \epsilon\!\geq\!\|\bm{y}-\bm{\hat{t}}_{\mathbb{D}}\|_2\!\geq\!\|\bm{A}_{\mathbb{D}}\bm{\alpha}\|_2\!-\!\|\bm{n}\|_2
    \geq \sigma_2(\bm{A}_{\mathbb{D}})\|\bm{\alpha}\|_2-\epsilon\ff{.}\label{eq:rev_tri}
\end{align}
Here $\sigma_2(\bm{A}_{\mathbb{D}})$ denotes the second largest singular value of $\bm{A}_{\mathbb{D}}$.
The last inequality holds 
whenever $\bm{A}_{\mathbb{D}}$ has full column rank, which can be verified to be true for ULAs and nested arrays \ff{when} $M\geq 2$. 
Rearranging \eqref{eq:rev_tri} yields
\begin{align}
    \sigma_2(\bm{A}_{\mathbb{D}})\leq \frac{2\epsilon}{\sqrt{\alpha^2+\hat{\alpha}^2}}\leq 2\frac{\epsilon}{\alpha}=2\rho^{-1/2},\label{eq:sigma_2_general}
\end{align}
where $\rho \triangleq (\alpha/\epsilon)^2$ is defined as the \ff{SNR}. Now, let $\bar{\omega} \triangleq \sin \hat{\theta}-\sin\theta$, and note that 
\begin{align*}
    \sigma_2^2(\bm{A}_{\mathbb{D}})
    =\sigma_2(\bm{A}_{\mathbb{D}}^{\HT}\bm{A}_{\mathbb{D}})
    =\sigma_2\Bigg(
    \begin{bmatrix}
        P&H_{\mathbb{D}}(\bar{\omega})\\
        H_{\mathbb{D}}^\ast(\bar{\omega})&P
    \end{bmatrix}
    \Bigg),
\end{align*}
since 
$\bm{a}_{\mathbb{D}}^{\HT}(\theta)\bm{a}_{\mathbb{D}}(\hat{\theta})=H_{\mathbb{D}}(\bar{\omega})$. \ff{The characteristic polynomial of $\bm{A}_{\mathbb{D}}^{\HT}\bm{A}_{\mathbb{D}}$ has two (non-negative) roots, the smaller one being}
\begin{align}
    \sigma_2^2(\bm{A}_{\mathbb{D}})=P-|H_{\mathbb{D}}(\bar{\omega})|. \label{eq:sigma_min}
\end{align}
This establishes the correspondence between the \emph{beampattern} $|H_{\mathbb{D}}(\bar{\omega})|$ and the second largest singular value of $\bm{A}_{\mathbb{D}}$. For convenience, instead of $\bar{\omega}\in(-2,2)$, we henceforth consider the \emph{wrap-around distance} $\Delta\in[0,1]$, defined as 
\begin{align}
    \Delta \triangleq 
    \min_{k\in\mathbb{Z}}|\sin\hat{\theta}-\sin\theta+2k|.\label{eq:Delta}
\end{align}
It can be verified that $|H_{\mathbb{D}}(\Delta)|=|H_{\mathbb{D}}(\bar{\omega})|$.  
\ff{Thus, by \labelcref{eq:sigma_2_general,eq:sigma_min,eq:Delta}}
\begin{align}
    \frac{1}{P}|H_{\mathbb{D}}(\Delta)|\geq 1 - \frac{4}{P\rho}. \label{eq:fcn_Delta_bound}
\end{align}
\cref{eq:fcn_Delta_bound} yields a useful \emph{necessary} condition that the angular error $\Delta$ of any solution to \eqref{p:psd} must satisfy. A key insight offered by \eqref{eq:fcn_Delta_bound} is that a low $\Delta$ is guaranteed if the SNR $\rho$ is within the {\em main lobe} of 
the beampattern. This implies that properly designed sparse arrays can achieve much lower angular error than the ULA, provided the side lobes of the sparse array are not too high. 
\cref{fig:beampattern} illustrates these observations in case of three arrays with $P=10$ sensors: the ULA $\mathbb{D}=\mathbb{U}_{10}$, dilated ULA $\mathbb{D}=3\mathbb{U}_{10}$, and nested array $\mathbb{D}=\mathbb{U}_5\cup (6\mathbb{U}_5+5)$. 
Note that \eqref{eq:fcn_Delta_bound} is independent of $N$. This is consistent with the notion that extrapolation cannot fundamentally decrease angular error. \ff{Finally, extensions to $K>1$ sources require ensuring that the solution to \eqref{p:psd} is of $\rank$ $K$, and characterizing the smallest singular value of a $2K\times 2K$ matrix.} 

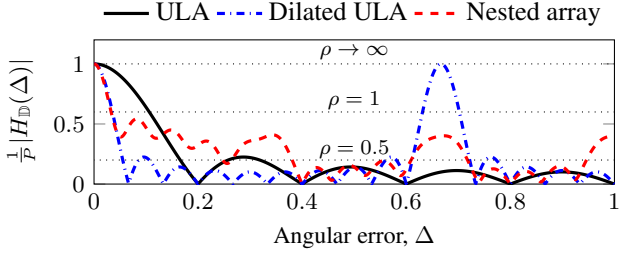
\begin{figure}
\newcommand{\numsens}{10}
\newcommand{\snra}{1}
\newcommand{\snrb}{0.5}
\pgfkeys{/pgf/number format/.cd,fixed,precision=1}
    \centering
    \begin{tikzpicture}
         \begin{axis}[width=8.5 cm,height=3.5 cm,ymin=0,ymax=1.2,xmin=0,xmax=1,ylabel={$\frac{1}{P}|H_{\mathbb{D}}(\Delta)|$},xlabel={Angular error, $\Delta$},title style={yshift=0 pt},xticklabel shift = 0 pt,xlabel shift = {0 pt},xtick style={draw=none},yticklabel shift=0pt,ylabel shift = 0 pt,ymode=linear,
         legend style = {at={(0.5,1.03)},anchor=south,draw=none,fill=none},legend columns=3]
                 \addplot[black,very thick] table[x=Delta,y=B_ula]{Data/beampatterns.dat};
                 \addlegendentry{ULA}
                 
                 \addplot[blue,dashdotted,very thick] table[x=Delta,y=B_sula]{Data/beampatterns.dat};
                 \addlegendentry{Dilated ULA}

                \addplot[red,dashed,very thick] table[x=Delta,y=B_nst]{Data/beampatterns.dat};
                 \addlegendentry{Nested array}
                 
                 \addplot[dotted] {1};
                \node[shift={(0cm,0.15cm)}] at (axis cs:(0.5,1) {\footnotesize{$\rho\to\infty$}};

                 \addplot[dotted] {1-4/\numsens*\snra^(-1)};
                \node[shift={(0cm,0.15cm)}] at (axis cs:(0.5,{1-4/\numsens*\snra^(-1)})
                {\footnotesize{$\rho=\snra$}};

                \addplot[dotted] {1-4/\numsens*\snrb^(-1)};
                \node[shift={(0cm,0.15cm)}] at (axis cs:(0.5,{1-4/\numsens*\snrb^(-1)}) {\footnotesize{$\rho=\snrb$}};
                 \end{axis}
    \end{tikzpicture}
    \caption{Beampattern $|H_{\mathbb{D}}(\Delta)|$ and SNR $\rho$ control angular error $\Delta$ of matrix completion. Horizontal dotted lines indicate the right hand side of \eqref{eq:fcn_Delta_bound} for different values of $\rho$. When $\rho$ is sufficiently large, $\Delta$ is determined by the { intersection of the corresponding horizontal line with the main lobe}, which is narrower for sparse arrays. However, low side lobes are also necessary to guarantee low $\Delta$. }\label{fig:beampattern}
\end{figure}

\section{Analytical upper bounds on angular error}\label{sec:upper_bounds}

\cref{fig:beampattern} showed that simply plotting the (irregular) beampattern of an arbitrary sparse array is a practical means of deducing an upper bound on the angular error $\Delta$. Since $|H_{\mathbb{D}}(\Delta)|$ can be a  complicated function of $\Delta$, analytically inverting it is challenging in general. However, for certain structured arrays, $|H_{\mathbb{D}}(\Delta)|$ becomes analytically tractable, which helps in characterizing \rd{the scaling  of $\Delta$} with parameters of interest such as $P$. Let $\Delta_\mathbb{D}$ denote the angle estimation error resulting from solving \eqref{p:psd} \rd{and} using any estimator from the class discussed \rd{in \cref{sec:signal_model}}. The following theorem provides upper bounds on $
$$\Delta_{\mathbb{D}}$ which reveal how fast the error decays with $P$.
\begin{thm}\label{thm:upper_bound}
    \bl{Let $P=2M$, where $M\geq 2$. If $\rho \geq 25/P$, then $\Delta_{\mathbb{U}_{P}} \leq \rd{1.2}/P$ and $\Delta_{\mathbb{S}} \leq 8/P^2$.}   
\end{thm}
\begin{proof}
\underline{ULA}: By basic properties of geometric series and trigonometric identities, $|H_{\mathbb{U}_m}(\Delta)| =\big|\frac{\sin(\pi m\Delta/2)}{\sin(\pi\Delta/2)}\big|.$
Following  $|\sin y|\leq \vert y \vert, \forall y$ and $\sin y\geq \frac{2}{\pi}y, y\in(0, \frac{\pi}{2}]$ \cite[p.~33]{mitrinovic1970analytic},
\begin{align}
|H_{\mathbb{U}_m}(\Delta)|
 \leq \min\bigg(m,\frac{1}{\Delta}\bigg) \quad 0<\Delta\leq 1.\label{ub_ula}
 \end{align}
 Since \eqref{eq:fcn_Delta_bound} holds for all feasible solutions, it is also true for the worst-case error $\Delta_{\mathbb{U}_{P}}$. Combining this with $\rho\geq 25/P$:
\begin{align*}
    &\frac{1}{P\Delta_{\mathbb{U}_{P}}}\geq \frac{1}{P}|H_{\mathbb{U}_P}(\Delta_{\mathbb{U}_{P}})|\geq 1-\frac{4}{P\rho}\geq 1-\frac{4}{25}\rd{=\frac{21}{25}}\\
    &\implies \Delta_{\mathbb{U}_{P}}\leq \frac{\rd{25/21}}{P}\leq \frac{\rd{1.2}}{P}.
\end{align*}

\underline{Nested Array}: The nested array can be written as $\mathbb{S}=\mathbb{S}_1\cup\mathbb{S}_2$ where $\mathbb{S}_1\triangleq \mathbb{U}_M$ and $\mathbb{S}_2\triangleq (M+1)\mathbb{U}_{M}+M$ are disjoint sets. Hence, by the triangle inequality,
\begin{align}
    |H_{\mathbb{S}}(\Delta)|
    &= \Bigg|\sum_{d_1\in \mathbb{S}_1}e^{j\pi d_1 \Delta}+\sum_{d_2\in \mathbb{S}_2}e^{j\pi d_2 \Delta}\Bigg|\nonumber\\
    &= \Bigg|\sum_{i= 1}^{M}e^{j\pi (i-1) \Delta}+e^{j\pi M \Delta}\sum_{i= 1}^{M}e^{j\pi(i-1)(M+1)\Delta}\Bigg|\nonumber\\
    &\leq|H_{\mathbb{U}_M}(\Delta)|+|H_{\mathbb{U}_M}((M+1)\Delta)|.\label{eq:nested_ub_1}
\end{align}
Alternatively, $\mathbb{S}\!=\!\mathbb{S}_1'\!\cup\!\mathbb{S}_2'$ where $\mathbb{S}_1'\!\triangleq\!\mathbb{U}_{M+1}$, $\mathbb{S}_2'\!\triangleq\!(M+1)\mathbb{U}_{M-1}+2M+1$, and $\mathbb{S}_1'\cap\mathbb{S}_2'=\emptyset$. Hence, similarly to \eqref{eq:nested_ub_1},
\begin{align}
     |H_{\mathbb{S}}(\Delta)| \leq |H_{\mathbb{U}_{M+1}}(\Delta)|+|H_{\mathbb{U}_{M-1}}((M+1)\Delta)|.\label{eq:nested_ub_2}
\end{align}
We proceed by showing that the worst-case angle error obeys $|H_{\mathbb{S}}(\Delta_{\mathbb{S}})|<0.84P$ when $\Delta_{\mathbb{S}}\in[\frac{2}{(M+1)M},1]$, which can be interpreted as upper bounding the highest side lobe level. We consider two subintervals:
\begin{enumerate}[label=(\roman*)]    
    \item\label{i:start} $\Delta_{\mathbb{S}}\in[\frac{2}{(M+1)M},\frac{1}{M+1})$: Denote $\Delta' = \Delta_{\mathbb{S}}(M+1)$. Thus, $\Delta'\in[\frac{2}{M},1)$, which by \eqref{ub_ula} implies that
    \begin{align*}
        |H_{\mathbb{U}_M}((M+1)\Delta_{\mathbb{S}})| =|H_{\mathbb{U}_M}(\Delta')|\leq \frac{1}{\Delta'}\leq \frac{M}{2}.
        \end{align*}
    Hence, $|H_\mathbb{S}(\Delta_{\mathbb{S}})| \leq M+\frac{M}{2}=\frac{3P}{4}$ by \labelcref{ub_ula,eq:nested_ub_1}.
    \item\label{i:mid} $\Delta_{\mathbb{S}}\!\in\![\frac{1}{M+1},1]$: Note that $\sin(\pi\Delta_{\mathbb{S}}/2)$ is a positive increasing function of $\Delta_{\mathbb{S}}$ when $\Delta_{\mathbb{S}}\!\in\![\frac{1}{M+1},1]$. Thus, $\frac{1}{\sin(\pi\Delta_{\mathbb{S}}/2)}\leq \frac{1}{\sin(\pi/(2(M+1))}$. Moreover, $\sin x\geq (1-\frac{x^2}{6})x, x\in[0,\frac{\pi}{2}]$ can be established using the Taylor series expansion \rd{of $\sin x$} \cite[Eq.~(3.1)]{klen2010jordan}.
    Applying these two facts to $|H_{\mathbb{U}_{M+1}}(\Delta_{\mathbb{S}})|$ yields
\begin{align}
     |H_{\mathbb{U}_{M+1}}(\Delta_{\mathbb{S}})|
     \!\leq\!\frac{1}{\sin(\frac{\pi/2}{M+1})}
     \!\leq\!\frac{(M+1)12/\pi}{6-(\frac{\pi/2}{M+1})^2}.\label{eq:ii}
\end{align}
Recalling assumption $M\geq 2$, we substitute $M\!=\!2$ in the denominator of \eqref{eq:ii}---an increasing function of $M$---to obtain $|H_{\mathbb{U}_{M+1}}(\Delta_{\mathbb{S}})|<0.67 (M+1)$. Thus by \eqref{eq:nested_ub_2} and \eqref{ub_ula}: $|H_{\mathbb{S}}(\Delta_{\mathbb{S}})|< 0.67(M+1)+M-1<1.67M< 0.84P$.
\end{enumerate}
In summary, if $\Delta_{\mathbb{S}}\!\geq\!\frac{2}{(M+1)M} $ then $\frac{1}{P}|H_{\mathbb{S}}(\Delta_{\mathbb{S}})|\!<\!\rd{\max(\frac{3}{4},0.84)}\\ =\!0.84$. \rd{Also note that $\frac{1}{P}|H_{\mathbb{S}}(\Delta_{\mathbb{S}})|\!=1\!>\!0.84$ when $\Delta_{\mathbb{S}}\!=\!0$. This shows that if $\frac{1}{P}|H_{\mathbb{S}}(\Delta_0)|\!>\!0.84$ for some $\Delta_0$, then it must hold that $\Delta_0\!\in\![0, \frac{2}{(M+1)M})]$.}
\cref{thm:upper_bound} now follows via contradiction: Let $ \rho \geq 25/P$. Suppose $\Delta_{\mathbb{S}} \geq \frac{2}{(M+1)M}$. However, by \eqref{eq:fcn_Delta_bound}: $0.84>\frac{1}{P}|H_{\mathbb{S}}(\Delta_{\mathbb{S}})|\geq 1-\frac{4}{P\rho}\geq 1-\frac{4}{25} = 0.84$, which is a contradiction. Hence, if $\rho \geq 25/P$ then $\Delta_{\mathbb{S}}<\frac{2}{(M+1)M}\leq \frac{2}{M^2}=\frac{8}{P^2}$. This completes the proof.
\end{proof}
\bl{\cref{thm:upper_bound} reveals an interesting fact: when the SNR is at least proportional to $\frac{1}{P}$, the \rd{upper bound on the} worst-case angle estimation error decays \rd{with $P$ at different rates for the ULA and the nested array}. The more favorable decay of the bound in case of the nested array---$\frac{1}{P^2}$ compared to $\frac{1}{P}$ for the ULA---can also be attributed to the narrower main lobe of its beampattern, as discussed in \cref{sec:beampattern}.} 
\section{Numerical experiments}

This section numerically validates the theory outlined in \cref{sec:beampattern,sec:upper_bounds}. We solve a well-known convex relaxation of \eqref{p:psd}, where instead of the rank, we minimize the trace\footnote{Since $\mathcal{T}(\bm{t})$ is PSD, minimizing the trace is equivalent to minimizing the nuclear norm (sum of singular values) or simply the first entry of vector $\bm{t}$.} of $\mathcal{T}(\bm{t})$ using the CVX toolbox \cite{grant2014cvx}. 
We obtain an estimate $\hat{\theta}$ of angle $\theta$ by applying root-MUSIC \cite{barabell1983improving} on $\mathcal{T}(\bm{\hat{t}})$. We repeat this experiment for $1000$ Monte Carlo trials, where both the ground truth $\sin\theta$, and the real and imaginary parts of the entries of noise vector $\bm{n}$ are drawn independently at random from a uniform distribution, such that $\sin\theta,\frac{\epsilon}{\sqrt{2P}}\Re\{n_i\},\frac{\epsilon}{\sqrt{2P}}\Im\{n_i\}\sim\mathcal{U}(-1,1), i = 1,2,\ldots,N$. We fix  $\alpha\!=\!1$. SNR is varied by only varying the noise level $\epsilon$.

\cref{fig:Delta_max_SNR} shows the angular error $\Delta$ as a function of SNR $\rho$ for the { extrapolated} ULA and nested array 
with $P=10$ sensors and 
$N=30$ (aperture of the nested array). \rd{An upper bound on} $\Delta$ (dashed curves) computed using \eqref{eq:fcn_Delta_bound} and the array beampatterns in \cref{fig:beampattern}, and the largest value of $\Delta$ observed over the Monte Carlo trials (solid curves) 
display similar trends and scaling with respect to the array geometries and $\rho$ 
. This demonstrates the utility of \eqref{eq:fcn_Delta_bound} for predicting the angular error of matrix completion for diverse array configurations. \cref{fig:Delta_max_SNR} shows an initial sharp transition from high to low error, which corresponds to the transition from the side lobe region to the mainlobe of the beampattern (cf.~\cref{fig:beampattern}). The { (extrapolated)} ULA seems to display an advantage in a narrow range of {low} SNR values, which can be attributed to its lower peak side lobe level. However, as SNR increases, the nested array achieves a consistently lower angular error thanks to its narrower main lobe, \emph{despite} the fact that the ULA is extrapolated to the same virtual aperture (since $N=30$). 

\begin{figure}
\newcommand{\numsens}{10}
    \centering
    \begin{tikzpicture}
         \begin{axis}[width=8 cm,height=4.0 cm,ylabel={\rd{Angular error,} $\Delta$},xlabel={SNR, $\rho$ (dB)},title style={yshift=0 pt},xticklabel shift = 0 pt,xlabel shift = {0 pt},xmin=-15,xmax=10,yticklabel shift=0pt,ylabel shift = 0 pt,ymode=log,
         legend style = {at={(0.5,1.03)},anchor=south,draw=none,fill=none},legend columns=4]
                 
        

                 \addplot[black,dashed,very thick,forget plot] table[x =SNR,y=Delta]{Data/Delta_max_SNR_ULA.dat};
                 
                 \addplot[red,dashed,very thick,forget plot] table[x =SNR,y=Delta]{Data/Delta_max_SNR_NA.dat};

                 \addplot[black, very thick] table[x=SNR,y=Delta]{Data/Delta_max_sim_SNR_ULA_unif.dat};
                 \addlegendentry{ULA (extrapolated)}
                 \addplot[red,very thick] table[x=SNR,y=Delta]{Data/Delta_max_sim_SNR_NA_unif.dat};
                 \addlegendentry{Nested array}

                 \end{axis}
    \end{tikzpicture}\vspace{-0.4cm}
    \caption{Angular error of matrix completion as a function of SNR $\rho$ (for $P=10$). The empirical error (solid) follows the trends of the upper bounds \rd{derived from \eqref{eq:fcn_Delta_bound}}  (dashed curves).}\label{fig:Delta_max_SNR}
\end{figure}
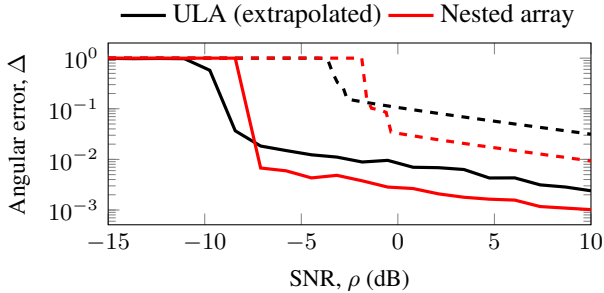

                 


\cref{fig:Delta_max_P} shows the scaling of the worst-case value of $\Delta$ as a function of the number of sensors $P$ when the SNR scales as $\rho=16/P$. 
The angular error of the nested array decays faster with $P$ compared to the { extrapolated} ULA both in case of the trends predicted by \cref{thm:upper_bound} (dashed lines), and the empirical maximum error over the Monte Carlo trials using trace minimization and root-MUSIC (solid lines). \cref{fig:Delta_max_P} also displays the approximate rates of decay estimated by a least-squares fit of linear model $\log\Delta =a\log P+b$ to the empirical curves. Over interval $P=4,6,\ldots,22$, a faster decay is observed experimentally---$P^{-1.6}$ (ULA) and $P^{-2.1}$ (nested array)---compared to \cref{thm:upper_bound} ($P^{-1}$ and $P^{-2}$). Whether this trend persists for larger values of $P$, or for other optimization algorithms and noise models is still an open question.
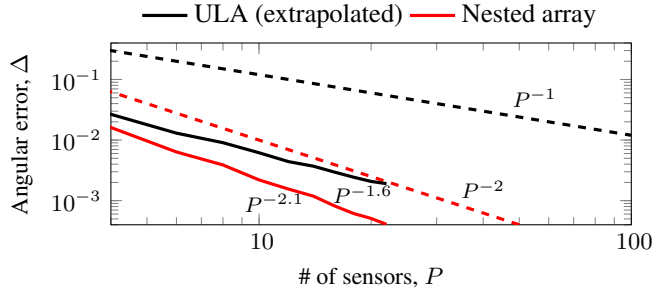
\begin{figure}
    \centering
    \begin{tikzpicture}
         \begin{axis}[width=8.5 cm,height=4.0 cm,ylabel={\rd{Angular error,} $\Delta$},xlabel={\# of sensors, $P$},title style={yshift=0 pt},xticklabel shift = 0 pt,xlabel shift = {0 pt},xmin=4,xmax=100,ymin=4e-4,ymax=4e-1,yticklabel shift=0pt,ylabel shift = 0 pt,xmode=log,log x ticks with fixed point,ymode=log,
         legend style = {at={(0.5,1.03)},anchor=south,draw=none,fill=none},legend columns=2]

                 \addplot[black,dashed,very thick,forget plot,domain = 2:100] {1.2/x};
                 

                 \addplot[red,dashed,very thick,forget plot,domain = 2:100] {1/x^2};
                 
                 
                 \addplot[black, very thick] table[x=P,y=Delta]{Data/Delta_max_sim_P_ULA_unif_lin.dat};
                 \addlegendentry{ULA (extrapolated)}

                 
                 \addplot[red, very thick] table[x=P,y=Delta]
                 {Data/Delta_max_sim_P_NA_unif_lin.dat};
                 \addlegendentry{Nested array}

               \node[shift={(0cm,0cm)}] at (axis cs:(55,5e-2) {\footnotesize{$P^{-1}$}};
                 \node[shift={(0cm,0.15cm)}] at (axis cs:(40,1e-3) {\rd{\footnotesize{$P^{-2}$}}};
                \node[shift={(0cm,0.1cm)}] at (axis cs:(19,9.5e-4) {\footnotesize{$P^{-1.6}$}};
                 \node[shift={(0cm,0.15cm)}] at (axis cs:(11,6e-4) {\rd{\footnotesize{$P^{-2.1}$}}};

                
                 \end{axis}
    \end{tikzpicture}\vspace{-.4cm}
    \caption{Angular error of matrix completion as a function of the number of sensors $P$ (for $\rho=\frac{16}{P}$). \rd{For $P\!\in[4,22]$,} the empirical curves (solid) decay slightly faster than the scaling laws $\frac{1}{P}$ and $\frac{1}{P^2}$ predicted by \cref{thm:upper_bound} (dashed lines).}\label{fig:Delta_max_P}
\end{figure}

\vspace{-0.2cm}
\section{Conclusions}
\vspace{-0.2cm}
We showed that the sparse array beampattern fundamentally controls the angular error of array interpolation based on matrix completion. Specifically, we derived upper bounds on the error of angle estimates revealed by solutions to a low rank Toeplitz completion problem, where noisy measurements of a single unknown (positive) source signal are observed. 
Using this insight, we proved that nested arrays can attain lower worst-case angle estimation error than ULAs (extrapolated to the same aperture) for comparable SNR. The theoretical findings were supported by numerical experiments that corroborate the advantages of sparse arrays compared to uniform arrays in noisy sample-starved regimes.

\bibliographystyle{IEEEtran}
\bibliography{references}
\end{document}